 \documentclass[preprint]{elsarticle}
\usepackage{amsfonts}
\usepackage{amssymb}
\usepackage{amsmath}

\usepackage{graphicx}
\usepackage{xcolor}
\usepackage{amssymb}
\usepackage{amsthm}
\usepackage{amsmath}

\theoremstyle{plain}
\newtheorem{theo}{Theorem}[section]
\newtheorem{lem}[theo]{Lemma}
\newtheorem{prop}[theo]{Proposition}
\newtheorem{cor}[theo]{Corollary}

\theoremstyle{definition}

\newtheorem{ex}[theo]{Example}

\newtheorem{remark}[theo]{Remark}

\journal{FFA}

\begin{document}

\begin{frontmatter}


\title{A transform approach to polycyclic and serial codes over rings}



\author[1,m]{Maryam Bajalan }  
\ead{mar.bajalan@gmail.com}
\fntext[m]{This work was completed  while this author visited the Institute of Mathematics of University of Valladolid (IMUVa) during Nov. 2020-- June 2021. She thanks the IMUVa for their kind hospitality.}
\author[2,v]{Edgar Mart\'inez-Moro} 
\ead{edgar.martinez@uva.es}
\fntext[v]{The second author is partially funded
by the Spanish Research Agency (AEI) under Grant PGC2018-096446-B-C21}
\author[3]{Steve Szabo}
\ead{Steve.Szabo@eku.edu}

\address[1]{Department of Mathematics, Malayer University, Hamedan, Iran}
\address[2]{Institute of Mathematics, University of Valladolid, Castilla, Spain}
\address[3]{Department of Mathematics \& Statistics, Eastern Kentucky University}

\begin{abstract}
In this paper, a transform approach is used for polycyclic and serial codes  over finite local rings in the case that the defining polynomials have no multiple roots. This allows us to study them in  terms of linear algebra and invariant subspaces as well as understand the duality in terms of the transform domain. We also make a characterization of when two polycyclic ambient spaces are Hamming-isometric.\end{abstract}


\end{frontmatter}



\section*{Introduction}
 Polycyclic codes over a local ring $R$  can be described as ideals on the ring $R[x]/\langle f(x) \rangle$
where $f$ is a polynomial in $R[x]$. They were introduced in \cite{SteveCycli} and are a generalization of cyclic and constacyclic codes which have been extensively studied in the literature. Polycyclic codes over finite fields have been studied from several points of view, see for example \cite{0dual,polyinv},    they have been also  studied over Galois rings \cite{Steve} and recently over chain rings in \cite{Fot20}.
In \cite{polyinv} the authors pointed out that it was worth to generalize their results from finite fields alphabets
to chain rings. In this paper we will make this generalization to finite local rings in the case that the polynomial defining the ambient space has simple roots (see Section~1 for a definition).  We will propose a transform approach that generalizes the classical Mattson-Solomon (Fourier)  transform  for finite fields, moreover, we show the relationship between the transform and the  annihilator duality for polycyclic codes introduced in \cite{0dual}. Note that this approach can be easily translated to the multivariable case as it is pointed out in the last section of this paper.

The outline of the paper is as follows. In  Section~\ref{S:1} we show those results on finite local rings, circulant matrices over rings and matrix diagonalization needed for our work. In Section~\ref{sec:DFT} we review the discrete Fourier Transform over rings as well as some facts on Vandermonde matrices over rings. Section~\ref{S:MS} is devoted to the description on the Mattson-Solomon transforms and its relationship with several inner products both in the original space and its transform image. The main result is Theorem~\ref{th:ann} that shows that all of them generate the same dual code. The generalization to finite local rings of the results in \cite{polyinv} can be found in Section~\ref{S:inv}. In Section~\ref{S:isometric}  we investigate when two different polycyclic definitions provide isomorphic and isometric coding ambient spaces. Finally in Section~\ref{sec:multivariable} we show how all the previous result can be generalized in the case of serial codes. 

\section{Preliminaries}
\label{S:1}

\subsection{Finite local rings}
We will show here selected results about local rings needed in the paper, for a complete account see \cite{Local}.
In this paper $R$ will denote a finite local ring of characteristic $q=p^r$ for a prime $p$ and a positive integer $r$, $\mathfrak m$
will denote the maximal ideal of $R$ and $\mathbb F_q = R/\mathfrak m$ the finite residue field
of $R$. It is well-known that $R$ is trivially complete and thus Hensel, i.e. every element of $R$ is nilpotent or a unit and  $\mathfrak m$ is a nilpotent ideal.  We denote by $\bar\cdot$ the natural polynomial ring morphism $\bar\cdot : R \mapsto (R/\mathfrak m)$ and abusing notation we will use it also for polynomial rings acting on the coefficients $\bar\cdot : R[x] \mapsto (R/\mathfrak m)[x]=\mathbb{F}_q[x]$. 

Let $\mathcal J$ denote the set of all polynomials $f$ in $R[x]$ such that
$\bar f$ has distinct zeros in the algebraic closure of $\mathbb{F}_q$,  a polynomial in $ \mathcal J$ has distinct zeros in local extensions of
$R$, $\mathcal R_f$
(where $f$ is monic) is a separable local extension
if and only if $f$ is an irreducible polynomial in $\mathcal J$, and the polynomials
in $\mathcal J$ admit unique factorizations into irreducible polynomials and a polynomial in $\mathcal J$ has no multiple roots in any
local extension of $R$. \textbf{Throughout the paper we will restrict to polynomials in $\mathcal J$ unless otherwise stated.} The following two lemmas will be helpful during the paper.

\begin{lem}[Azumaya's Lemma] Let $f$ be a monic polynomial in
$R[x]$. Then $\mathcal R_f=I_1\bigoplus I_2$ where $I_1$ and $I_2$ are ideals in $\mathcal R_f$
if and only if there exist monic coprime polynomials $h$ and $g$ in $R[x]$
with $f=gh$ and $I_1= \langle g\rangle / \langle f\rangle $, $I_2= \langle h\rangle / \langle f\rangle$.
\end{lem}

 An element $e$ of the ring $\mathcal R_f$  is called an idempotent if $e^2=e;$ two idempotents $e_1, e_2$ are said to be orthogonal if $e_1 e_2=0$ and an idempotent is said to be primitive if it is non-zero and cannot be written as the sum of non-zero orthogonal idempotents. A set $\{e_1,...,e_r\}$ of elements of $\mathcal R_f$ is called a complete set of idempotents if $\sum_{i=1}^r e_i=1$.
If $\{e_1,...,e_r\}$ is a complete set of pairwise orthogonal idempotents, we have  that $\mathcal R_f = \bigoplus _{i=1}^r \mathcal R_f e_i $.

\begin{lem}[Theorem 3.2 in \cite{idempotents}] \label{lem:idem} Let  $R$ be a finite local commutative ring and $f$ be a monic polynomial in $R[x]$ such that $f = \prod_{i=1}^{r} f_i$ is the {unique} factorization of $f$ into
a product of monic primary pairwise coprime polynomials. 
 The ring $\mathcal R_f$ admits a unique complete set of primitive pairwise orthogonal idempotents $\{e_1, e_2,...,e_r\}$ given by
\begin{equation}
    e_i= v_i(x) \hat{f}_i(x),~\text{where}~ v_i(x) \in \mathcal R_f ~\text{and}~ \hat{f}_i=\frac{f}{f_i}.
\end{equation}
Moreover $e_i R[x] \cong \frac{R[X]}{\langle f_i\rangle}$ and $\mathcal R_f=\bigoplus_{i=1}^r e_i R[x].$
\end{lem}

\subsection{Circulant matrices}

We will denote by $\mathcal{M}_n( R)$ the set of $n\times n$ matrices over the local ring $R$. If $\mathrm{deg}f(x)=n$,  $E_f\in \mathcal{M}_n(R)$ will be the companion matrix associated with $f(x)=x^n -\sum_{i=0}^{n-1} f_i x^i$,
\begin{equation}
    E_f=\left( \begin{array}{ccccc}
     0    & 1 &0  & \cdots& 0  \\
      0   & 0& 1 & \cdots & 0\\   \vdots&  \vdots &   \vdots& \ddots  & \vdots \\
      0&  0 & 0 & \cdots  &1\\
     f_0 & f_1 &  f_2 & \cdots  & f_{n-1}\end{array}\right).
\end{equation}
Consider the usual matrix multiplication in $\mathcal{M}_n(R)$ and the ordinary product in $\mathcal{R}_f.$ Consider the basis $\mathcal B=\{1,x,x^2,\ldots,x^{n-1}\} $ for $\mathcal R_f$
and let the map $\rho_f: \mathcal R_f\to R^n$ send a polynomial to coefficients of $x^i$. The map $M: \mathcal{R}_f \rightarrow \mathcal M_n(R)$ defined  by
$$M(g(x))=
\begin{bmatrix}
 \rho_f(g(x))\\
 \rho_f(xg(x))\\
 \vdots\\
 \rho_f(x^{n-1}g(x))
\end{bmatrix}
$$
is the regular representation of elements $\mathcal R_f.$ If we denote the image of $M$ by  $\mathcal{M}_n(R,f),$ then $M:\mathcal R_f\to \mathcal{M}_n(R,f) $ is a ring isomorphism.
Clearly $M(x)=E_f$ and hence  set $\{\mathrm{Id}, E_f, E_f^2,\ldots, E_f^{n-1}\}$ is a basis for  $\mathcal{M}_n(R,f)$, in fact the   elements of  $\mathcal{M}_n(R,f)$ are   linear combination of powers of the companion matrix $E_f$.
 This isomorphism has been extensively studied in  \cite{Vel13}.  Note  that elements $\mathcal{M}_n(R,f)$  are called Barnett matrices in \cite{Vel13},   $f(x)$-circulants in  \cite{CH95} or  polycirculant matrices in \cite{Poly}.
The following characterization of the subrings of  $\mathcal{M}_n(R)$  being images of such an ismorphims can be found in \cite{Vel13}.
\begin{lem}[ Theorem 2.1~\cite{Vel13}]\label{centralizer}
A subring $S$ of $\mathcal{M}_n(R)$ is of the form $\mathcal{M}_n(R,f)$ if and only if $S=C_{\mathcal{M}_n(R)}(E_f)$, the centralizer of he matrix $E_f$ in $\mathcal{M}_n(R)$.
\end{lem}
\noindent This fact plays a central role for the diagonalization of commuting matrices in the field case \cite{CH95}.
We will denote by $\mathcal M_{1, n}(R,f)$  the set of all $1\times n$ matrices $[a_0,\,a_1\ldots\, a_{n-1}]$ endowed with the following multiplication 
\begin{equation}\label{multiplication}
    [a_0,\,a_1\ldots\, a_{n-1}]\cdot[b_0,\,b_1\ldots\, b_{n-1}]=[a_0,\,a_1\ldots\, a_{n-1}]M(b),
\end{equation}
where $b=b_0+b_1x+\ldots+b_{n-1}x^{n-1}.$
Since every element of $\mathcal M_n(R,f) $ is determined by its first row and polynomial $f(x)$,  the map $\varphi: \mathcal M_n(R, f)\rightarrow (\mathcal M_{1, n}(R,f),\cdot)$,  which sends every matrix to the first row is a ring isomorphims.
\section{The Discrete Fourier Transform  over commutative rings}\label{sec:DFT}
We assume that the reader is familiar with the  Discrete Fourier Transform (DFT) over finite fields and its applications to cyclic codes (see \cite{MS77} for example).  Suppose that   $\xi$  is a  primitive $N^{th}$ root of unity in a field $\mathbb F,$ i.e,  $\xi^N=1$ and $\xi\neq 1$ for $i=1,\ldots, N-1.$ For any integer $j,$
\begin{equation} \label{cron}
\sum\limits_{i = 0}^{N - 1} {{\xi ^{ij}}}=\begin{cases}
			N, &  j=0(\text{mod N}),\\
            0, & \text{otherwise}.
		 \end{cases}
\end{equation}
and the DFT of length $N$ generated by $\xi$ is the  mapping $DFT_{\xi}$ from $\mathbb F^N$ to $\mathbb F^N$ defined by $B=DFT_\xi(b),$ where $B_i=\sum_{n=0}^{N-1}b_n\xi^{in}$ for $i=0,1,\ldots ,N-1$ or equvalently $B=bM_{\xi},$ where 
$$M_{\xi}=\begin{bmatrix}
1 & 1 & 1 &  \ldots & 1\\
1 & \xi & {\xi}^2 & \ldots & {\xi}^{N-1}\\
1 & {\xi}^2 & {\xi}^{2.2} & \ldots & {\xi}^{2(N-1)}\\
\vdots &\vdots &\vdots & \ldots &\vdots\\
1 & {\xi}^{N-1}& {\xi}^{(N-1)2} &\ldots & {\xi}^{(N-1)(N-1)}
\end{bmatrix}.$$
Thus $M_{\xi}$ is a    Vandermonde matrix with determinant $\prod\limits_{j=1}^{N-1}\prod\limits_{i=1}^{j-1}(\xi^j-\xi^i)$  which is non-zero and hence   $M_{\xi}$ is non-singular  and the  DFT is always invertible. The  inverse transform of DFT is given by
\begin{equation}\label{invers}
    b_n=\frac{1}{N}\sum\limits_{i=0}^{N-1}B_i\xi^{-in},\quad \text{for}\quad   n=0,1,\ldots, N-1.
\end{equation}
Note that a matrix over a ring is non-singular if and only if its determinant is a unit in the ring \cite{WI93}. Moreover, a  product of elements of a ring is unit if and only if each element is unit. Therefore the following  theorem holds.
\begin{theo}\label{DFTring}(DFT   over rings \cite[Theorem 10]{A12})
If $\xi$ is a primitive $N^{th}$ root of unity in a commutative ring $R.$ Then the  DFT from $R^N$ to $R^N$ defines an invertible mapping  whose invers is given by the  Equation \eqref{invers} if and only if $\xi^k-1$ is a unit of $R$ for $k=1,2,\ldots,N-1.$
\end{theo}
For example, $\xi=2$ is a primitive $4^{th}$ root of unity in $\mathbb Z_{15}$ but $\xi^2-1$ is not unit. So for $\xi=2,$  DFT of lenght $N=4$ is not invertible. There are  some results for DFT over $\mathbb Z_m,$ (Number Theory Transform in  \cite{M1998}).
It follows from above theorem that if $\xi$ generates an invertible  DFT of length $N$ in ring $R$ and $L(>1)$ is a divisor $N,$ then $\xi^{\frac{N}{L}}$ also generates an invertible  DFT of length $L$ in $R$ (see \cite{M1998}).  
\subsection{Vandermonde matrices over  commutative rings}
 Let $R$ be the local ring with extension $R'$ and $R$-algebra morphism $\gamma: R \rightarrow R'.$ A matrix $M\in \mathcal{M}_n(R)$ is diagonizable over $R'$ if there are matrices $V, D\in \mathcal{M}_n(R')$  such that $D$ is  diagonal and $V^{-1}\gamma (M)V=D$,  see  \cite{La13}.  From now on, for convenience we simply write $V^{-1}MV=D$. 
Consider $f(x)\in R[x]$ such that  $f(x)=\prod_{i=1}^n (x-\alpha_i) \in R^\prime[x]$, i.e. $f(x)$ splits in  $R^\prime$. Also
consider the following Vandermonde matrix  
$$
V=V(\alpha_1, \ldots ,\alpha_{n})=
\begin{bmatrix}
1 & 1 & \ldots & 1\\
\alpha_1 & \alpha_2 & \ldots &\alpha_n\\
\alpha_1^2& \alpha_2^2&\ldots & \alpha_n^2\\
\vdots&\vdots&\ldots&\vdots\\
\alpha_1^{n-1}& \alpha_2^{n-1}&\ldots&\alpha_n^{n-1}
\end{bmatrix}.
$$
For $i=1,\ldots, n,$ denote $i^{th}$ column of $V$ as  $V_i$  and for  $j=1,\ldots, n,$ and  given $g\in\mathcal R_f$  denote $j^{th}$ row of $M=M(g)$ by  $M_j.$ We know that entries of $M_j$ are the coefficients of $x^{j-1}g(x)(\mod{f})$ and henceforth  $M_jV_i=\alpha_i^{j-1}g(\alpha_i).$ Therefore $MV_i=g(\alpha_i)V_i,$ that is  $V_i$ is  the eigenvector and $g(\alpha_i)$ is the eigenvalue of $M$ in $R'.$ So $MV=\mathrm{diag}[g(\alpha_1),g(\alpha_2),\ldots,g(\alpha_{n})]V.$
Now  If $V$ is non-singular then $V^{-1}MV=\mathrm{diag}[g(\alpha_1),g(\alpha_2),\ldots,g(\alpha_{n})],$ and hence $M$ is diagonalized by $V.$ 
We know that $\det V=\prod_{j=1}^{n-1}\prod_{i=1}^{j-1}(\alpha_j-\alpha_i)$ is in local ring $R'.$ It is well-known that for a local ring  $\alpha_j-\alpha_i(i\neq j)$ is a unit if and only if $\bar \alpha_j\neq \bar \alpha_i$, see \cite{NS00}. 
So $V$ is non-singular if and only if  $\bar \alpha_j\neq \bar \alpha_i,$ for all $i\neq j.$
Note that if $f(x)\in \mathcal J\subseteq R[x]$   then $V$ is non-singular. Moreover, for a    non-singular matrix $A$ over a local ring then the homogeneous system $Ax=0$ has a unique solution \cite[Lemma 2.1]{NS00}.  The following result provides us $V^{-1}$.


\begin{lem}[\cite{R19}]
If $V^T$ is non-singular, then $(V^T)^{-1}=(w_{ij})$ is given by 
$$(w_{ij})=(-1)^{i+j}\frac{S_{n-i,j}}{\prod\limits_ {l<k}^{n}(\alpha_k-\alpha_l)}$$
with $l=j$ or $k=j$ and $S_k=S_k(\alpha_1,\ldots,\alpha_n)=\sum\limits_{1\leqslant i_1<\ldots<i_k\leqslant n}\alpha_{i_1}\ldots \alpha_{i_k},$ and $S_0(\alpha_1,\alpha_2,\ldots,\alpha_{n})=1$ and $S_{k,j}=S_k(\alpha_1,\ldots,\alpha_{j-1},\alpha_{j+1},\ldots,\alpha_n).$
\end{lem}
\begin{ex}
Let $n=3$ and
$f(x)=x^3+5x+3\in\mathbb Z_9[x].$ Then $f(x)=(x-1)(x-12)(x-23)\in\mathbb Z_{27}[x].$
 Moreover $\det(V)=16$ is unit in $\mathbb Z_{27}.$ We have 
 \begin{equation*}
 V^{-1}=
      \begin{bmatrix}
21 & 8 & 26\\
19 & 6 & 2\\
15 & 13 & 26\\
\end{bmatrix}.
 \end{equation*}
  \end{ex}

\section{Mattson-Solomon transform and  polycyclic codes over rings}\label{S:MS}

From now on we will be concerned with univariate polycyclic codes defined as ideals of the ambient space $\mathcal{R}_f$, where  $f(x)$ a monic   polynomial in $\mathcal J\subseteq R[x]$.
For a local ring the diagonalizing of $M$ is unique up to permutation of diagonal entries. So let $\{\alpha_1,\ldots ,\alpha_n\}$ be a fixed ordering of roots of $f(x)$ in extension ring $R'$ of $R.$
The map
\begin{equation}\label{eq:ms}\begin{array}{cccc}
   MS_f:&  (\mathcal{R}_f, \cdot ) &\longrightarrow & (R^\prime[x]/\langle f(x)\rangle, \star) \\[0.5em]
     & g(x) & \mapsto & \displaystyle  \sum_{i=1}^{n} g(\alpha_i) x^{i-1},
\end{array}
\end{equation}
is a ring homormorphism, 
where $\cdot$ denotes ordinary polynomial multiplication modulo  $f(x)$ and   $\star$ denotes the component-wise multiplication or Schur product.  We will call  the map in (\ref{eq:ms})  the \emph{Mattson-Solomon transform} with respect to the polynomial $f(x)$. Indeed in the case $f(x)=x^n-1$ we recover the Fourier transform  in the previous section. Since $f(x)\in \mathcal J$,  the  Vandermonde matrix $V$ is non-singular and hence the    homomorphism  $MS_f$ is injective. Take $V^{-1}=(u_{ij})$ and   $g(x)=g_0+g_1x+\ldots+g_{n-1}x^{n-1}\in \mathcal R_f$ and denote $g=(g_0,g_1,\ldots,g_{n-1}).$ If we  denoted by $B=MS_f(g),$ then $B_i=g(\alpha_i)$ and  the inverse formula is given by
\begin{equation*}
   g_{j-1}=\sum\limits_{k=1}^{n}B_{k-1}u_{kj},\quad 1\leqslant j\leqslant n,
\end{equation*}
 i.e.  
$
  g(x)=\sum\limits_{j=1}^n \sum\limits_{k=1}^{n}B_{k-1}u_{kj}x^{j-1}.  
$

\begin{ex}[Example 1 Cont.] Let $g(x)=g_0+g_1x+g_2x^2\in \mathbb Z_9[x]$ and $MS_f(g)=(B_0,B_1,B_2)$ then the inverse of the Mattson-Solomon transform is
 \begin{equation*}
    g_0=21B_0+19B_1+15B_2,\, \, g_1=8B_0+6B_1+13B_2,\,\, g_2=26B_0+2B_1+26B_2.
      \end{equation*}
\end{ex}

Given two polynomials $g_1(x),g_2(x)\in \frac{R'[x]}{\langle f(x)\rangle}=\mathcal R'_f,$ 
we define the $\star$ inner product  as  
\begin{equation}\label{dual1}
    \langle g_1(x),g_2(x)\rangle_{\star}=(g_1)_0(g_2)_0+\ldots+ (g_1)_{n-1}(g_2)_{n-1}=(g_1\star g_2)(1).
\end{equation}
Since  $\langle g(x),  x^i \rangle_{\star}=g_i$ for $i=1,\ldots, n,$  the inner product is   non-degenerate.
Let $\mathcal C\subseteq \mathcal R_f$ be a polycyclic code.  The dual of $\mathcal C$ w.r.t. $\star$, denoted by  ${\mathcal C}^{\perp_{\star}},$ is define as
\begin{equation}
  {\mathcal C}^{\perp_{\star}}=\{ h(x)\in \mathcal R_f\mid  \langle MS(g),MS(h)\rangle_{\star}=0 \hbox{ for all } g(x)\in \mathcal C \}.  
\end{equation}
\noindent We  also   define  an  inner-product on $\mathcal R_f$   by
\begin{equation}\label{inner}
    \langle g_1(x), g_2(x) \rangle_{\mathrm{tr}}=  \text{trace}\,(M(g_1g_2)),\quad  g_1(x), g_2(x) \in \mathcal R_f, 
\end{equation}
and denote $\mathcal C^{\perp_{tr}}=\{g\in \mathcal C\mid \langle g(x),h(x)\rangle_{tr}=0 \quad \text{for all}\quad  h\in \mathcal C \}$.  Let $k(x)\in \mathcal R_f$ and $\langle k(x),x^i\rangle_{tr}=0$ for all $i=0,1,\ldots,n-1.$
Then $\langle MS_f(k(x),MS_f(x^i))\rangle_{\star}=0.$ Thus $\alpha_1^ik(\alpha_1)+\ldots+\alpha_n^ik(\alpha_n)=0$ and hence $(k(\alpha_1),\ldots,k(\alpha_n))V=0.$ Since $V$ is non-singular, then the linear homogeneous system has the unique solution $k(\alpha_1)=\ldots=k(\alpha_n)=0.$ So $MS_f(k(x))=0,$ i.e. $k(x)=0.$ Henceforth the trace inner product is non-degenerate. 

If $\pi_i$ denotes the projection of $\mathcal R_f$ onto the coefficient of $x^i$ for $i=0,\ldots, n-1$, the trace map is defined in \cite{Local} as $\mathrm{tr}:\mathcal R_f\rightarrow R$  is given by
$
    \mathrm{tr}(g)=\sum \pi_i(x^ig).
$ It is clear that the trace inner product of $g_1,g_2$ is equal to  the  trace map  of $g_1g_2$.
\begin{prop}\label{dual preserving}  Let $g_1(x), g_2(x) \in \mathcal R_f$, then 
$$ \langle g_1(x), g_2(x) \rangle_{\mathrm{tr}}=0 \Longleftrightarrow \langle MS_f(g_1),\, MS_f(g_2)\rangle 
_{\star}=0 $$
\end{prop}
\begin{proof}
We know $\gamma(M(g_1g_2))$ is similar to $\mathrm{diag}[(g_1g_2)(\alpha_1),\ldots,(g_1g_2)(\alpha_n)].$ So 
\begin{equation*}
\begin{split}
\gamma(\mathrm{trace}(M(g_1g_2))) & = \mathrm{trace}(\gamma(M(g_1g_2)))\\
 & = (g_1g_2)(\alpha_1)+\ldots +(g_1g_2)(\alpha_n)\\
 &=\langle MS_f(g_1),\, MS_f(g_2)\rangle_{\star}.
\end{split}
\end{equation*}
\end{proof}

 Now consider the following inner product
on $\mathcal R_f,$
\begin{equation}
    \langle g_1(x), g_2(x) \rangle_{(0)}=  g_1g_2(0),\quad  g_1(x), g_2(x) \in \mathcal R_f. 
\end{equation} it is a non-degenerate symmetric bilinear form if $f_0\neq 0.$ The dual $\mathcal C^{\perp_0}$ of  code $\mathcal C$ is just its annihilator dual in  \cite{0dual,Fot20}. Since the matrix $V$ is non-singular the Mattson-Solomon transform is an  injective morphism and thus it is clear that for any ideal $\mathcal C$ we have $$\mathrm{Ann}(\mathcal C)= \mathcal C^{\perp_{MS}}=\{g\in \mathcal R_f\mid MS(g)\star MS(c)=0\,\hbox{ for all } c\in\mathcal C \}.$$
We have this result that identifies the dualities in the transform domain.
\begin{theo}\label{th:ann} Let $R$ be a finite local ring. If $f\in\mathcal J\subseteq R[x]$  and $\mathcal C$ is a code in $\mathcal R_f$ we have that
$$\mathcal C^{\perp_\mathrm{tr}}=\mathcal C^{\perp_\star}=C^{\perp_0}=\mathcal C^{\perp_{MS}}= \mathrm{Ann}(\mathcal C).$$
\end{theo}
\begin{proof}
From the discussion above it is clear that $\mathcal C^{\perp_\mathrm{tr}}=\mathcal C^{\perp_\star}$ and $C^{\perp_0}=\mathcal C^{\perp_{MS}}= \mathrm{Ann}(\mathcal C)$. Moreover $\mathrm{Ann}(\mathcal C)\subseteq \mathcal C^{\perp_\mathrm{tr}}$ is straight forward. Let's prove the other direction.
Suppose there exist an element $g\in \mathcal C^{\perp_\mathrm{tr}}$ such that $g\notin \mathrm{Ann}(\mathcal C)=\mathcal C^c$, thus $g\in \mathcal C\cap \mathcal C^{\perp_\mathrm{tr}}$ and hence $g=r\cdot \sum_{j=1}^k e_{i_j}$ where  $\sum_{j=1}^k e_{i_j}$ is the idempotent generating $\mathcal C$. Consider now any other element $s\in \mathcal R_f$, then $s=\sum_{i=1}^t s_i e_{i}$ and $g\cdot s=\sum_{j=1}^k rs_{i_j} e_{i_j}$. Hence $0=\langle g,s\rangle_{\mathrm tr}= \langle r,s\rangle_{\mathrm tr}$ for all $s$ in $\mathcal R_f$, thus $r=0$ since the trace inner product is non-degenerate and therefore $g=0$.  
\end{proof}

\section{Polycyclic codes as invariant spaces}\label{S:inv}


Let $f_1,f_2,\ldots f_r$ be pairwise coprime monic polynomials over $R,$ $f=f_1f_2\ldots f_r$ and $\hat{f_i}=\frac{f}{f_i}.$ There exists $a_i,b_i\in R$ such that $a_if_i+b_i\hat{f_i}=1.$ Let $e_i=b_i\hat{f_i}+\langle f(x)\rangle$ as in Lemma~\ref{lem:idem}.    Let us 
define the set $U_i\subseteq \mathcal M_{1,n}(R,f)$ as $U_i=ker f_i(E_f).$
\begin{prop}\label{fundamental}\
\begin{enumerate}
 \item $R_fe_i=Ann_{\mathcal R_f} (f_i).$
 \item $c\in\mathcal R_fe_i$ if and only if $e_ic=c.$
   
      \item The matrix $M(e_i)=e_i(E_f)$ is the generator matrix of  the polycyclic code $\mathcal R_f e_i.$
      \item $M(e_i)$ is invariant under multiplication by  the companion matrix $E_f$ for all $i=1,\ldots r$ and they are pairwise orthogonal  idempotent matrices.
      \item The image of a  polycyclic code under $M$ is an invariant ideal under multiplication by $E_f.$
    \item $U_i\cong \mathcal R_fe_i.$
\end{enumerate}
\end{prop}
\begin{proof}\
\begin{enumerate}
    
    \item It follows from the annihilator definition.
\item It follows from $a_if_i+b_i\hat{f}_i=1$ and Part $1.$
\item  The rows of $M(e_i)$ are the coefficients of $e_i(x),\,xe_i(x),\ldots,x^{n-1}e_i(x).$
\item $M(\mathcal R_fe_i)$ is an ideal in $\mathcal M_n(R,f)$ and we know that $\{\mathrm{Id}, E_f, E_f^2,\ldots, E_f^{n-1}\}$ is a basis for  $\mathcal{M}_n(R,f).$ Moreover, since $M$ is an isomorphism,  for all $i\ne j$ we have $M(e_i)\ne M(e_j),$  $M(e_i)M(e_i)=M(e_i^2)=M(e_i),$ and $M(e_i)M(e_j)=M(0)=0$.

\item It is clear.
\item 
It is enough to prove that for a fixed $i,$ $\rho_f(\mathcal R_f e_i)=U_i.$ Let $k=re_i\in \mathcal R_fe_i $ for some $r\in \mathcal R_f.$ Since $e_if_i=0,$ we have $k(E_f)f_i(E_f)=0,$ and hence $\varphi(k(E_f))f_i(E_f)=0,$ i.e. $\rho_f(k)\in U_i.$ Conversely, let $[a_0 \ldots a_{n-1}]\in U_i.$ Denote $f_i(x)=b_0+b_1x+\ldots+b_{n-1}x^{n-1}.$ Applying definition \eqref{multiplication}, we have $[a_0 \ldots a_{n-1}].[b_0\ldots b_{n-1}]=0.$ If we denote $a(x)=a_0+a_1x+\ldots+a_{n-1}x^{n-1},$ 
$$a(x)f_i(x)=\rho_f^{-1}([a_0\ldots a_{n-1}].[b_0 \ldots b_{n-1}])=0.$$
Therefore $a(x)\in Ann (f_i),$ i.e. $\rho_f^{-1}[a_0\ldots a_{n-1}]\in \mathcal R_fe_i$ and the proof is complete.
\end{enumerate}
\end{proof}

Taking into account the previous proposition and since the ring $\mathcal R_f$ admits a unique complete set of primitive pairwise orthogonal idempotents $\{e_1, e_2,...,e_r\},$ polycyclic codes over rings  decompose  into some minimal polycyclic codes corresponding to each one of them  and the following result follows.

\begin{theo}\label{th:companion}
Let $E_f$ be companion matrix with minimal polynomial $f$ and $f=f_1\ldots f_r$ decomposes in to pairwise coprime monic irreducible polynomials.   Then
$$ \mathcal R_f \cong U_{1}\oplus\ldots\oplus U_{r},$$ where $U_{i}=Ker f_i(E_f)\subseteq \mathcal M_{1,n}(R,f).$
\begin{enumerate}
   \item Each of the above summands is an  indecomposible polycyclic code with $E_f$-invariant image in $\mathcal M_n(R,f).$ 
   \item Each polycyclic code $\mathcal C\subseteq \mathcal R_f$ can be seen as  a direct sum of  $U_{i}$ for some indices $i.$
   \item If a polycyclic code $\mathcal C$ decomposes as  $\mathcal C\cong \bigoplus _{i\in I} U_i$ then $$\mathcal C^{\perp_{tr}} \cong  \left( \bigoplus_{i\in I} U_i\right)^c= \bigoplus_{i\notin I} U_i$$
\end{enumerate}
\end{theo}

Note that this result generalizes the  decomposition in \cite{polyinv} based on 
the Primary Decomposition Theorem in linear algebra over finite fields.

\begin{remark}[BCH-like Bounds]
Mattson-Solomon transform can be  a great tool for understanding BCH-like bounds that have been established for different types of (chain) rings (see for example \cite{GS13,LXG16}) or based on invariant spaces for polycyclic codes over fields \cite{polyinv}. Note that the minimum distance of a linear general code over $R$ is the same as the one of its socle \cite[Proposition~5]{Nechaev}. Note that in general, we can not state that the minimum distance of a   code $\mathcal C$ is
equal to the minimum distance of the code $\bar{\mathcal C}$, since in general $d(\mathcal C)\leq d(\bar{\mathcal C})$. However, if they are Hensel's lifts of codes over $\mathbb F_q$ (see \cite{Serial} for a characterization)  we have the equality and therefore all classical bounds on distances for   codes over fields
(Bose–Ray-Chaudhuri–Hocquenghem, Hartmann–Tzeng, Roos, etc.) also apply to
their Hensel's lifts.
\end{remark}

\section{Isometric ambient spaces}
\label{S:isometric}
Assume that a finite ring $R$ is equipped with a weight $w.$ Linear codes $\mathcal C, \mathcal D\subseteq R^n$  are called isometric if there exists an $R$-linear isomorphism $\phi:C\to D$ which $w(\phi(c))=w(c)$ for all $c\in C.$ In the literature, codes $\mathcal C, \mathcal D$ are called isometrically equivalent. The
MacWilliams Extension Theorem, one of the most poweful theorems, states that the map $\phi: \mathcal C\rightarrow \mathcal D$ between  linear codes over $R$ is the Hamming-isometry if and only if it is  a monomial transformation, i.e for every $c\in \mathcal C$ there is a monomial matrix $M_c$ such that $\phi(c)=cM_c.$ Notice that every Hamming-isometry  is a homogeneous-isometry and vice versa, \cite{wood}. \\
\begin{theo}\label{isovel}(\cite{Vel13} Theorem 3.1)
Let  $h(x)=x^n-h_{n-1}x^{n-1}-\cdots -h_1x-h_0$ be a polynomial in $R[x]$ of the same degree of $f$. If there exists a polynomial   $\omega\in \mathcal R_f$ such that  $h\left( \omega\right)=0 \in  \mathcal R_f$,  and  $\det (W)$ is a unit in $R,$ where  $W=
\begin{bmatrix}
  \rho_f(\omega^0) \\
   \rho_f(\omega^1) \\
 \rho_f(\omega^2) \\
  \vdots \\
   \rho_f(\omega^{n-1}) 
\end{bmatrix}$  then 
\begin{equation}\begin{array}{cccc}
   \theta:&  \mathcal M_{1, n}(R,h) &\longrightarrow &  \mathcal M_{1, n}(R,f) \\[0.5em]
     & \rho_h (x)& \mapsto & \rho_f(\omega),
\end{array}
\end{equation}
is an isomorphism which is the identity in  $R$ (where $R$ is identify with $\rho_h(r)$, $r\in \ \mathcal R_f$ a constant polynomial). 
\end{theo}
\begin{remark}\label{reiso}
To construct such a polynomial $h(x)$ as described  in the first statement of Theorem \ref{isovel}, choose $\rho_f(\omega)\in \mathcal M_{1, n}(R,f)$ such that $\det W$ is a unit element in $R.$ Now assume $[h_0\, h_1\ldots h_{n-1}]=\rho_f(\omega^n) W^{-1}.$
\end{remark}

\begin{ex}\label{20}\
\begin{enumerate}
    \item Let $R=\mathbb Z_4,$ $f(x)=x^3-2x^2-x-1$ and $h(x)=x^3-x^2-1$ are polynomials in $R[x].$ If $\omega =1+x^2\in \frac{R[x]}{\langle f(x)\rangle}$ then $\theta: \frac{R[x]}{\langle h(x)\rangle}\rightarrow \frac{R[x]}{\langle f(x)\rangle}$ is an isomorphism since $h(\omega)=0$ and $\det{W}=1.$ Note that it is  not a Hamming isometry, because $\theta(x^2+x+1)=3x+1.$ 
    \item Let $R=\mathbb Z_4,$ $f(x)=x^4-3x-1$ and $h(x)=x^4-2x^2-x-3$ are polynomials in $R[x].$ If $\omega =3x+1\in \frac{R[x]}{\langle f(x)\rangle}$ then $\theta: \frac{R[x]}{\langle h(x)\rangle}\rightarrow \frac{R[x]}{\langle f(x)\rangle}$ is an isomorphism since $h(\omega)=0$ and $\det{W}=1.$ Note that this one is  not isometry, because $\theta(x^2)=x^2+2x+1.$ \end{enumerate}
 
\end{ex}

We know that isometrically equivalent  linear codes have both the same   algebraic structure and  distance properties thus it will be nice to know when two polycyclic ambient spaces are isometric or not. In \cite{Dinh15}, Dinh and Li classify all isometrically equivalent classes of constacyclic codes and only study representatives of equivalent classes.  The previous example shows that   $  \mathcal R_f$ and $ \mathcal R_h$ are not necessarily isometrically equivalent for different polynomials $f,h$ of the same degree even if they are isomorphic. In the rest of this section we will  describe when a polycyclic   ambient space   $\mathcal R_f$   is isometrically equivalent to another one. 
Notice  that  the isomorphism $\theta$ is an  isometry if and only if  $W$ is a monomial matrix.
  \begin{prop}\label{mono}
 With the notation above,
$W$ is a monomial matrix if and only if either $f(x)=x^n-f_0$ and $\omega_ix^i,$ where $f_0,\,\omega_i\in R^*$ and $(n,i)=1$ or $f(x)=x^n-f_1x$ and $\omega=\omega_jx^j,$ where $f_1,\,\omega_j\in R^*$ and $(n-1,j)=1$
   
\end{prop}
\begin{proof}
\
We know that $W$ is monomial if and only if 
\begin{equation}\label{Mon}
\{\omega^k\,:\, 1\leqslant k\leqslant n-1\}=\{a_ix^i\,:\, 1\leqslant i\leqslant n-1, a_i\in R^*\}.
\end{equation}
Suppose that $W$ is monomial. Let $w=w_0+w_1x+\ldots w_{n-1}x^{n-1}$ and $f(x)=x^n-\lambda(x).$ Since $W$ is monomial, $w$ can not be  the sum of two terms or more. So there is $i, 1\leqslant i\leqslant n-1,$ such that $\omega=\omega_ix^i$ and $\omega_i\in R^*.$ Moreover, if $\lambda(x)$ is the sum of two terms or more, then there is $k,1\leqslant k\leqslant n-1,$ such that $\omega^k$ is the sum of two terms, a contradiction. So $\lambda(x)=f_0,\, f_0\in R^*$ or $\lambda(x)=f_1x,\, f_1\in R^*.$ Notice  that if $\lambda(x)=f_tx^t, t\geqslant 2,$ then $f(x)$ is not in $\mathcal J$. In  the case $f(x)=x^n-f_1x,$ let $(n-1,j)=k>1.$ Then there is $t_1,\,t_2<n-2$ such that $kt_1=j$ and $kt_2=n-1.$ We obtain 
$\omega ^{t_2}=x^{jt_2}=x^{kt_1t_2}=x^{nt_1}x^{-t_1}=1,$ a contradiction with \eqref{Mon}. With a similar discussion  in the other case, we prove $(n,i)=1$.\\
Conversely, let $f(x)=x^n-f_1x$ and $\omega=\omega_jx^j.$ The left side of incusion in \eqref{Mon} is trivial. For the other direction,  by contradiction assume that there are $k_1,\,k_2<n$ such that $k_1\neq k_2$ and $\omega^{k_1}=\omega^{k_2}$. So
$x^{k_1j}=x^{(n-1)+k_2j}$ and hence $(k_1-k_2)j=n-1.$ A similar discussion occurs for the case $f(x)=x^n-f_0.$
\end{proof}
\begin{cor}\label{isometry}
\
\begin{enumerate}
    \item Let $f(x)=x^n-f_0$ and $\omega=\omega_ix^i,$ where $(n,i)=1$ and $f_0,\,\omega_i\in R^*.$ Then 
    $\mathcal M_{1, n}(R,f)$ and 
    $\mathcal M_{1, n}(R,x^n-\omega_i^nf_0)$ are isometric.
    \item Let $f(x)=x^n-f_1x$ and $\omega=\omega_jx^j,$ where $(n-1,j)=1$ and $f_1,\,\omega_j\in R^*.$ Then 
    $\mathcal M_{1, n}(R,f)$ and 
    $\mathcal M_{1, n}(R,x^n-\omega_j^{n-1}f_1^jx)$ are isometric.
   \end{enumerate}
\end{cor}
\noindent As a corollary we can recover the result \cite[ Theorem 4.3]{BGG12} as follows.
\begin{cor}
Let $n$ be an integer and there is $\lambda \in R^*$ such that $n^{th}$ root of $\lambda$ is an element in $R^*.$ Then $\lambda-$constacylic code of lenght $n$ is isometrically  equivalent to the cyclic code of lenght $n.$
\end{cor}
\begin{proof}
choose integer $i<n$ such that $(n,i)=1.$ We know  there is an element $\omega_i\in R^*$ such that $\omega_i^n=\lambda.$ Then $\mathcal M_{1, n} (R,x^n-\lambda)$ and $\mathcal M_{1, n}(R,x^n-1)$ are isometric.
\end{proof}
\begin{ex}
 Let $f(x)=x^6-f_1x$ and $\omega=\omega_4x^4,$ where
    $f_1,\omega_4\in R^*.$ Then 
    $$
    W=
\begin{bmatrix}
1 & 0 & 0 & 0 & 0 & 0\\
0 & 0 & 0 & 0 & \omega_4 & 0\\
0 & 0 & 0 & \omega_4^2f_1 & 0 & 0\\
0 & 0 & \omega_4^3f_1^2 & 0 & 0 & 0\\
0 & \omega_4^4f_1^3 & 0 & 0 & 0 & 0\\
0 & 0 & 0 & 0 & 0 & \omega_4^5f_1^3
\end{bmatrix}.
$$
We have $\omega^6=\omega_4^6f_1^4x^4$ and hence by Remark \ref{reiso}
 $$   [h_0\,\,\,h_1\,\ldots\, h_6]=[0\,\,\,0\,\,\,0\,\,\,0\,\,\,\,\omega_4^6f_1^4\,\,\,\,0]W^{-1}=[0\,\,\,\,\omega_4^5f_1^4\,\,\,0\,\,\,0\,\,\,0\,\,\,0].$$ Then $h(x)=x^6-\omega_4^5f_1^4x.$
  
\end{ex}


\section{Multivariable serial codes and transform domain}\label{sec:multivariable}

From now on we will assume that $R$ is a chain ring. A multivariable serial code   over $R$ is an ideal of the ring ${R[x_1,\ldots,,x_r]}/{\langle f_1(x),\ldots,f_r(x_r)\rangle},$ where $ f_i(x)\in \mathcal J$  for all $i=1,\ldots,r$, for an account on serial codes see \cite{Serial}. In this section we will propose a transform approach to those codes defining it duality. For the sake of simplicity all results in this section will be proved for $r=2$ and can be straight forward  worked out for $r>2$. Let $f_1(x), f_2(x)$ be polynomials in $R[x]$ of degree $n_1,n_2,$ respectively, we will
denote  the  multivariable ring ${R[x_1,x_2]}/{\langle f_1(x_1),f_2(x_2)\rangle}$ by $\mathcal R_{f_1,f_2}.$  There is an extension $R'$ of $R$ such that $f_1$ and $f_2$  splits over $R'$. Let $\{\alpha_1,\ldots,\alpha_{n_1}\}$ be a fixed ordering of roots of $f_1$ in $R'$ and $\{\beta_1,\ldots,\beta_{n_2}\}$ be that of $f_2$ in $R'$.  


The tensor product of two $R$-modules $A,B$ is an  $R$-module denoted by $A\otimes B$ with multiplication $(a\otimes b)(c\otimes d)=ac\otimes db.$ If $A,B$ are free $R$-modules with basis $X_1,X_2,$ respectively, then $\{x_1\otimes x_2: x_1\in X_1, x_2\in X_2\}$ is a basis of $A\otimes B.$
If $A,B$ are free $R$-modules and  $I$ be a submodule of free $R$-module $A\otimes B,$ then  there are submodules $I_1\in A$ and $I_2\in B$ such that $I=I_1\otimes I_2.$ 
If $f: A\to B$ and $g: A'\to B' $ be $R$-module isomorphisms, then $f\otimes g: A\otimes B\to A'\otimes B'$ defended as $(f\otimes g)(x\otimes y)=f(x)\otimes g(y)$ is an $R$-module isomorphism. Tensor product over direct sum of modules is distributive.\\ 
Recall that the tensor product of matrices  $A$ of size ${m\times n}$ and $B$ of size ${p \times q}$(denote by $\otimes$) is $mp\times n q$ matrix $A\otimes B=(a_{i,j}B).$ If $A$ and $B$ are square matrices, then  $\mathrm{det}(A\otimes B)=(\mathrm{det} A)^{m}(\mathrm{det} B)^p$ and $\mathrm{tr}(A\otimes B)=(\mathrm{tr}\,A)(\mathrm{tr}\,B).$  Also  for matrices $A,A',B,B'$ we have that   $(A\otimes  B)(A'\otimes B')=(AA')\otimes(BB')$ mixes the ordinary matrix product and tensor product(mixed-product property). For more information on the tensor product of modules and matrices the reader can refer to \cite{Jacobson}.

 \begin{ex} This example is  for clarifying the influence of the basis one can choose. 
  Let the polynomials $f(x)=f_1+f_2x$ in $R[x]$ and $g(y)=g_0+g_1y+g_2y^2$ in $R[y]$  and also consider the basis $\beta=\{1,y,y^2,x,xy,xy^2\}$ on  $\mathcal R_{f,g}.$ By computing the representation matrix of elements  of $\mathcal R_{f,g}$ we see that they are related to companion matrices $E_f$ and $E_f$ as follows:
  \begin{enumerate}
      \item the representation matrix $x$ is $E_f\otimes \mathrm{Id}_3$
       \item the representation matrix $y$ is $\mathrm{Id}_2\otimes E_g$
        \item the representation matrix $xy$ is $E_f\otimes E_g$
         \item the representation matrix $xy^2$ is $E_f\otimes E_g^2$
          \item the representation matrix $y^2$ is $\mathrm{Id}_2\otimes E_g^2$
  \end{enumerate}
   Thus $\beta'=\{\mathrm{Id}_2\otimes \mathrm{Id}_3,\mathrm{Id}_2\otimes E_g,\mathrm{Id}_2\otimes E_g^2, E_f\otimes \mathrm{Id}_3, E_f\otimes E_g, E_f\otimes E_g^2\}$ is its associated basis for the  representation matrices.
   Note that if we 
    choose now another basis $\beta=\{1,x,y,y^2,xy,xy^2\},$ the representation matrix of  elements is not equal to tensor product of $E_f,E_g,$ like above,  but after  permutation on rows of  the representation matrix we see that both  will be equal. 
  \end{ex}
Consider basis $\theta=\{\theta_1,\ldots,\theta_{n_1n_2}\}$  for the  multivariable ring $\mathcal R_{f_1,f_2}$ described in above example. A matrix representation of each element of $\mathcal R_{f_1,f_2}$ can be computed with respect to the basis $\theta$. Recall that  notations $\mathcal M_{n_1}(R,f_1), \mathcal M_{n_2}(R,f_2)$ are used for matrix representations of rings $R[x_1]/f_1(x_1), R[x_2]/f_2(x_2),$ respectively. By Theorem \ref{centralizer} and mixed-product property, it is obvious that $\mathcal M_{n_1}(R,f_1)\otimes \mathcal M_{n_1}(R,f_2)$ is commutative.  Let $\rho$ denote a map from $\mathcal R_{f_1f_2}$ onto coefficients $x_1^ix_2^j.$ Then the map $M: \mathcal R_{f_1f_2}\to \mathcal M_{n_1}(R,f_1)\otimes \mathcal M_{n_1}(R,f_2)$ is defined by 
$$M(k(x_1,x_2))=
\begin{bmatrix}
\rho(\theta_1k(x_1,x_2))\\
\rho(\theta_2k(x_1,x_2))\\
\vdots\\
\rho(\theta_{n_1n_2}k(x_1,x_2))
\end{bmatrix}$$
is the regular representation of $\mathcal R_{f_1,f_2}$ with respect to $\theta.$ In fact $M$ maps element  $k(x,y)=\sum_i^{n_1-1}\sum_j^{n_2-1}k_{ij}x_1^ix_2^j$ to $\sum_i^{n_1-1}\sum_j^{n_2-1}k_{ij}E_{f_1}^i\otimes E_{f_2}^j.$ Clearly $M$ is an isomorphism. Note also that this fact arises from the fact that 
$R[x_1,x_2]/\langle f_1(x_1),f_2(x_2)\rangle \cong R[x_1]/\langle f_1(x_1)\rangle\otimes R[x_2]/\langle f_2(x_2)\rangle$, see \cite{Cazaran}, moreover it is a principal ideal ring if both $\mathcal R_f$ and $\mathcal R_f$ are principal ideal rings, see \cite{Kelarev}.

Let $V_{f_1}=V(\alpha_1,\ldots,\alpha_{n_1})$  and $V_{f_2}=V(\beta_1,\ldots,\beta_{n_1})$ be the  Vandermonde matrices associated to $f$ and $g$ respectively. We know that the image of the companion matrices $E_{f_1},E_{f_2}$ by $\gamma$ is diagonalizable by $V_{f_1}, V_{f_2},$ respectively.  Denote  $V=V_{f_1}\otimes V_{f_2}$ and suppose that $k(x_1,x_2)\in \mathcal R_{f_1,f_2}$ is an arbitrary element.
For simplicity  we  take  $M(K(x_1,x_2))=\gamma(M(K(x_1,x_2)))$, then 

\begin{align*}
M(K(x_1,x_2))V&=\sum\limits_i^{n_1-1}\sum\limits_j^{n_2-1}k_{ij}(E_{f_1}^i\otimes E_{f_2}^j)(V_{f_1}\otimes V_{f_2})\\         
&= \sum\limits_i^{n_1-1}\sum\limits_j^{n_2-1}k_{ij}(E_{f_1}^iV_{f_1})\otimes (E_{f_2}^jV_{f_2}) \\   
&= \sum\limits_i^{n_1-1}\sum\limits_j^{n_2-1}k_{ij}(V_{f_1}\,\mathrm{diag}[\alpha_1,\ldots,\alpha_{n_1}])\otimes(V_{f_2}\mathrm{diag}[\beta_1,\ldots,\beta_{n_2}])\\
&=\sum\limits_i^{n_1-1}\sum\limits_j^{n_2-1}k_{ij}(V_f\otimes V_g)(\mathrm{diag}[\alpha_1,\ldots,\alpha_{n_1}]\otimes \mathrm{diag}[\beta_1,\ldots,\beta_{n_2}])\\
&=V\sum\limits_i^{n_1-1}\sum\limits_j^{n_2-1}k_{ij}\mathrm{diag}[\alpha_1,\ldots,\alpha_{n_1}]\otimes \mathrm{diag}[\beta_1,\ldots,\beta_{n_2}]\\
&=V(\mathrm{diag}[k(\alpha_1,\beta_1),\ldots,k(\alpha_1,\beta_{n_2}),\ldots,k(\alpha_{n_1},\beta_1),\ldots,k(\alpha_{n_1},\beta_{n_2})]).
\end{align*}
Recall  that  since $f_1, f_2\in \mathcal J$ then $V_{f_1}$ and $ V_{f_2}$ are non-singular. Therefore, since  $\det(V)=(\det V_{f_1})^{n_1}(\det V_{f_2})^{n_2}$ then $V$ is non-singular.
Hence  $M(K(x_1,x_2))$ is diagonalized by $V$ and its eigenvalues are related to the roots of $f$ and $g$ as above.

Now we are able to define multivariable Mattson-Solomon transform for serial codes  as follows. 
\begin{equation}\begin{array}{cccc}
   MS_{f_1,f_2}:&  (\mathcal{R}_{f_1,f_2}, \cdot ) &\longrightarrow & (R'[x_1,x_2]/\langle f_1(x_1),f_2(x_2)\rangle, \star) \\[0.5em]
     & k(x_1,k_2) & \mapsto & \displaystyle 
     \sum\limits_i^{n_1}\sum\limits_j^{n_2}k(\alpha_i,\beta_j)x_1^{i-1}x_2^{j-1},
     \end{array}
\end{equation}
where $.$ denotes ordinary polynomial multiplication modulo  $f_1(x_1),f_2(x_2)$ and   $\star$ denotes the component-wise multiplication. We
define the  inner product $\langle\,,\rangle_{\star}$ over $R'[x_1,x_2]/\langle f_1(x_1), f_2(x_2)\rangle$ as   in \eqref{dual1}. The dual of   the polycyclic code $\mathcal C$ with respect to this inner product is denoted by $\mathcal C^{\perp_{\star}}$,  and furthermore, we define trace inner product on $\mathcal R_{f_1,f_2}$ as 
\begin{equation}
      \langle k_1(x_1,x_2),k_2(x_1,x_2)\rangle_{tr}=\mathrm{trace}(M(k_1(x_1,x_2)k_2(x_1,x_2))).
\end{equation}
Since $V$ is non-singular, then trace inner product is non-degenerate. Denote the trace  dual of the  multivarible  code $\mathcal C$ by $\mathcal C^{\perp_{tr}}.$
 The following result is proven as  Proposition~\ref{dual preserving}.
  \begin{prop}  Let $k_1(x_1,x_2), k_2(x_1,x_2) \in \mathcal R_{f_1,f_2}$, then 
$$ \langle k_1(x_1,x_2), k_2(x_1,x_2) \rangle_{tr}=0 \Longleftrightarrow \langle MS_{f_1,f_2}( k_1(x_1,x_2),),\, MS_{f_1,f_2}( k_2(x_1,x_2),)\rangle 
_{\star}=0 $$
\end{prop}

\begin{lem}[\cite{Serial}]\label{Multi-idem}
There is a complete set  
of central orthogonal idempotents in $\mathcal R_{f_1,f_2}.$
\end{lem}
The proof result follows from Proposition~3.7 and Remark~5 in \cite{Serial} where an explicit construction of such idempotents is made.

Taking into account the previous result it is easy to proof the following theorem  following the proof of Theorem~\ref{th:ann}.

\begin{theo}Let $R$ be a finite chain ring. If $\mathcal C$ is a  multivariable code in $\mathcal R_{f_1f_2}$ we have 
$$\mathcal C^{\perp_\mathrm{tr}}=\mathcal C^{\perp_{\star}}= \mathrm{Ann}(\mathcal C).$$
\end{theo}

Assume that  $\{e_k\}_{k\in K}$ is the complete set  
of centraly orthogonal idempotents in $\mathcal R_{f_1,f_2}.$ Also assume that  $f_1=\prod_{i\in I}p_i$ and $f_2=\prod _{j\in J}q_j$ are pairwise coprime decompositions of  $f_1,f_2$ and  $\{e_i\}_{i\in I}$ and $\{e_j\}_{j\in J}$ are the complete set  of centraly orthogonal idempotents in $\mathcal R_{f_1}$ and $\mathcal R_{f_2},$ respectively.  
Let us set  $U_{i}\subseteq \mathcal M_{1,n_1}(R,f_1)$ as $U_{i}=\ker  p_i(E_{f_1})$ and $U_{j}\subseteq \mathcal M_{1,n_2}(R,f_2)$ as $U_{j}=\ker q_j(E_{f_2}).$ We know that that $\mathcal R_{f_1f_2}\cong \mathcal R_{f_1}\otimes \mathcal R_{f_2},$ thus
by CRT theorem, Propositon \ref{fundamental} and the distributivity of tensor product over direct sum we have 
\begin{align*}
  \bigoplus\limits_{k\in K}  \mathcal R_{f_1,f_2}e_k & \cong \mathcal R_{f_1,f_2}\\
 &  \cong  (\bigoplus\limits_{i\in I} \mathcal R_{f_1}e_i)\otimes  (\bigoplus\limits_{j\in J} \mathcal R_{f_2}e_j)\\
 &  =\bigoplus\limits_{i\in I}\bigoplus\limits_{j\in J}(R_{f_1}e_i\otimes R_{f_2}e_j). 
\end{align*}
Note that 
a primitive central idempotent in  $A\otimes B$ is the tensor product  of  primitive  central idempotents of $A$ and $B$, and therefore 
  with above notations, we have following results  similar to Proposition~\ref{fundamental} and Theorem~\ref{th:companion}.

\begin{prop}\

\begin{enumerate}
 
 \item  For $e_k$ there is $p_i$ and $q_j$ such that $\mathcal R_{f_1,f_2}e_k \cong Ann (p_i)\otimes Ann (q_j).$
 \item $c\in\mathcal R_{f_1,f_2}e_k$ if and only if $e_kc=c.$
         \item  $M(e_k)$ is the generator matrix of the multivariable serial code $\mathcal R_{f_1,f_2} e_k.$
      \item $M(e_k)$ is invariant under multiplication by  matrices  $E_{f_1}^n\otimes E_{f_2}^m$, where $0\leqslant n\leqslant n_1-1$ and $0\leqslant m\leqslant n_2-1.$ Moreover $\{M(e_k)\}_{k\in K}$ are idempotent matrices and pairwise orthogonal.
      \item The image of  a multivariable serial code under $M$ is an invariant ideal under multiplication by all $E_{f_1}^n\otimes E_{f_2}^m$, where $0\leqslant n\leqslant n_1-1$ and $0\leqslant m\leqslant n_2-1.$
    \item $U_{i}\otimes U_{j}\cong \mathcal R_{f_1,f_2}e_k$ for some $i,j.$
\end{enumerate}
\end{prop}

    
\begin{theo}
We have 
$$ \mathcal R_{f_1,f_2} \cong \bigoplus\limits_{i\in I}\bigoplus\limits_{j\in J}(U_i\otimes U_j)$$ 
\begin{enumerate}
   \item Each of the above summands   is an  indecomposible multivariable  serial code with $E_{f_1}^n\otimes E_{f_2}^m$-invariant image  in $\mathcal M_{n_1}(R,f_1)\otimes M_{n_2}(R,f_2)$ for all $0\leqslant n\leqslant n_1-1 , 0\leqslant m \leqslant n_2-1.$ 
   \item Each multivariable serial  code $\mathcal C\subseteq \mathcal R_{f_1,f_2}$ can be seen as  a direct sum of  $U_i\otimes U_j$ for some indices $i,j.$
   \item If a serial code $\mathcal C$ decomposes as  $\mathcal C\cong \bigoplus _{i\in I_1}\bigoplus_{j\in J_1} U_i\otimes U_j, $ for $I_1\subseteq I,J_2\subseteq J,$ then $$\mathcal C^{\perp_{tr}} \cong  \left( \bigoplus_{i\in I_1}\bigoplus_{j\in J_1} U_i\otimes U_j\right)^c= \bigoplus_{i\notin I_1}\bigoplus_{j\notin J_1} U_i\otimes U_j$$
\end{enumerate}
\end{theo}

\begin{remark}
Note that during this section we only needed the ring $R$ to be a chain ring in those parts where those results of the construction of the idempotents in \cite{Serial} where needed.
\end{remark}

Now we return to the general case where R is  local ring.
Assume that  polynomials $f_1(x_1),h_1(x_1)\in R[x_1]$ have the same degree  of $n_1$ and $f_2(x_2),h_2(x_2)\in R[x_2]$ have the same degree  of $n_2$. The main question is that when multivariable codes over rings $\mathcal R_{f_1,f_2}$ and  $\mathcal R_{h_1,h_2}$ are isometric. We know that 
$$\mathcal R_{f_1,f_2}\cong \mathcal M_{n_1}(R,f_1)\otimes \mathcal M_{n_2}(R,f_2)\cong\mathcal M_{1,n_1}(R,f_1)\otimes \mathcal M_{1,n_2}(R,f_2).$$
Moreover, we know that the  
tensor product of two submodules  of free modules 
$\mathcal M_{1,n_1}(R,f_1), \mathcal M_{1,n_2}(R,f_2),$  is a submodule of their  tensor product. Now 
applying Corollary \ref{isometry}, we conclude that
\begin{prop}
\end{prop}
\begin{enumerate}
    \item
    
    Let $f_1(x_1)=x_1^{n_1}-\lambda_1$ and $\omega_1=\omega_ix_1^i$ where $(n_1,i)=1,$ $\lambda_1,\omega_i\in R^*.$\\
    Also let $f_2(x_2)=x_2^{n_2}-\lambda_2$ and $\omega_2=\omega_jx_2^j$ where $(n_2,j)=1,$ $\lambda_2,\omega_j\in R^*.$
    Then $R[x_1,x_2]/\langle f_1(x_1),f_2(x_1)\rangle$ and 
    $R[x_1,x_2]/\langle x_1^{n_1}-\omega_i^{n_1}\lambda_1,\,x_2^{n_2}-\omega_j^{n_2}\lambda_2 \rangle$ are isometrc. 
    \item 
     Let $f_1(x_1)=x_1^{n_1}-\lambda_1$ and $\omega_1=\omega_ix_1^i$ where $(n_1,i)=1,$ $\lambda_1,\omega_i\in R^*.$\\
    Also let $f_2(x_2)=x_2^{n_2}-\lambda_2x_2$ and $\omega_2=\omega_jx_2^j$ where $(n_2-1,j)=1,$ $\lambda_2,\omega_j\in R^*.$
    Then $R[x_1,x_2]/\langle f_1(x_1),f_2(x_1)\rangle$ and 
    $R[x_1,x_2]/\langle x_1^{n_1}-\omega_i^{n_1}\lambda_1,\,x_2^{n_2}-\omega_j^{n_2-1}\lambda_2^jx_2 \rangle$ are isometric. 
    
    \item  Let $f_1(x_1)=x_1^{n_1}-\lambda_1x_1$ and $\omega_1=\omega_ix_1^i$ where $(n_1-1,i)=1,$ $\lambda_1,\omega_i\in R^*.$\\
    Also let $f_2(x_2)=x_2^{n_2}-\lambda_2x_2$ and $\omega_2=\omega_jx_2^j$ where $(n_2-1,j)=1,$ $\lambda_2,\omega_j\in R^*.$
    Then $R[x_1,x_2]/\langle f_1(x_1),f_2(x_1)\rangle$ and 
    $R[x_1,x_2]/\langle x_1^{n_1}-\omega_i^{n_1-1}\lambda_1^ix_1,\,x_2^{n_2}-\omega_j^{n_2-1}\lambda_2^jx_2 \rangle$ are isometric. 

\end{enumerate}

\section{Conclusions}
In the present paper, we have developed a transform approach to polycyclic codes  under the hypothesis that the polynomial defining the ambient space is multiplicity free  which is equivalent, in
the cyclic codes case, to the coprimality of the length and the alphabet size. We have also extended that approach to multivariable serial codes under an equivalent hypothesis.
The main open
problem is to derive a similar approach for the repeated root case, at least for the case when the ambient space is a principal ideal ring \cite{PIR}.

\bibliographystyle{plain}
 
\bibliography{sample.bib}

\end{document}